\documentclass[12pt]{article}

\usepackage[margin=1in]{geometry}
\usepackage{amsmath, amssymb, amsthm}
\usepackage{mathtools}
\usepackage{booktabs}
\usepackage{array}
\usepackage{hyperref}
\usepackage{xcolor}
\usepackage{enumitem, csquotes}
\usepackage{authblk}

\DeclareUnicodeCharacter{2500}{---}   
\DeclareUnicodeCharacter{2013}{--}    
\DeclareUnicodeCharacter{2014}{---}   
\newtheorem{theorem}{Theorem}[section]
\newtheorem{lemma}[theorem]{Lemma}
\newtheorem{corollary}[theorem]{Corollary}
\newtheorem{proposition}[theorem]{Proposition}
\newtheorem{definition}[theorem]{Definition}

\newcommand{\iv}[1]{\bigl[#1\bigr]}   
\newcommand{\sig}{\sigma}  
\newcommand{\Sin}{S^{+}} 
\newcommand{\wnd}{w_{\mathrm{nd}}}             
\newcommand{\OmegaW}{\Omega}  
\newcommand{\Sout}{S^{-}}                       

\newcommand{\DS}{D_S^*}
\newcommand{\DSfull}{D_S}

\newcommand{\Splus}{S^+}
\newcommand{\Sminus}{S^-}

\newcommand{\du}{d_u}
\newcommand{\DSstar}{D_S^{*}}

\title{%
	\textbf{Antimagic Labeling of Central Digraphs
	}\\[2pt]
	\textbf{ Constructed from Loops}
}

\author[1]{Raja Rawat}

\affil[1]{%
	Department of Data Science, Chandigarh University,
	Uttar Pradesh 209859, India \\
	\normalfont\small\href{mailto:20rawraj@gmail.com}{\texttt{20rawraj@gmail.com, raja.l100436@culko.in}}
}
\date{}
\begin{document}
	
	\maketitle

\begin{abstract}
Let $(T, \circ)$ be a loop of order $k \geq 3$ and $\DSstar$ the loopless
central digraph obtained from the associated central digraph $\DSfull$
by removing the $k$ self-loops at the diagonal vertices.
We construct an explicit arc-labeling $f$ of $D_S^*$ ordering all
arcs lexicographically by $(\text{tail}, \text{head})$ and assigning
labels $1, \ldots, m$ in that order and prove that $f$ is an
antimagic labeling. Consequently, $\DS$ is antimagic for every loop of order
$k \geq 3$, a property that depends only on $k$ and is invariant
across all non-isomorphic loops of the same order.

\medskip

\medskip


\end{abstract}	
\section{Introduction}\label{sec:intro}

The subject of graph labelling is a major one in graph theory and involves assigning integers to the elements of a graph. The components can be vertices, edges, or arcs, the assignments being governed by certain constraints.
Since Rosa's seminal 1967 work on graceful~labellings \cite{Rosa1967}, numerous labelling paradigms have been investigated. Among these, \emph{antimagic labeling} has attracted sustained attention because of the elegance of its definition and pecularity.

An antimagic labelling of a graph~$G$ with $m$ edges is a bijection
$f\colon E(G) \to \{1, 2, \ldots, m\}$ such that the vertex sums where the
sum at a vertex~$v$ is the total of the labels on all edges incident
to~$v$ are pairwise distinct. Hartsfield and Ringel~\cite{HR1990}
introduced the concept in 1990 and conjectured that every connected graph other
than $K_2$ (complete graph of order $2$) is antimagic. Alon, Kaplan, Lev, Roditty, and Yuster~\cite{kaplan}
established antimagicness for graphs whose minimum degree is at least
$c \log n$ for an absolute constant~$c$. Cranston, Liang, and
Zhu~\cite{CranLiangZhu2015} settled the conjecture for regular graphs of odd
degree, and Eccles~\cite{Eccles2016} proved the result for all graphs of
sufficiently large average degree.  

The directed analogue of antimagic labelling was introduced by Hefetz, M\"utze, and Schwartz~\cite{HMS2010}, who defined an antimagic labeling of a directed
graph~$D$ having adjacency matrix $A$ with $m$ arcs as a bijection $f\colon A \to \{1, \ldots, m\}$ such that the \emph{oriented vertex sums}
$w(v) = \Splus(v) - \Sminus(v)$, difference between the sum of labels on
arcs entering~$v$ and the sum on arcs leaving~$v$ are pairwise distinct.
They also conjectured that every connected graph admits an antimagic orientation and proved partial cases; the conjecture was subsequently verified for regular graphs of even degree by Li, Song, Wang, Yang, and Zhang~\cite{LSWYZ2019}, for biregular bipartite graphs by Shan and Yu~\cite{ShanYu2017}, and for all graphs with minimum degree at least $33$ by Shan~\cite{shan}.

A largely independent line of research concerns \emph{central groupoids} and
their associated digraphs.  A central groupoid is a set~$S$ equipped with a
binary operation~$*$ satisfying the identity $(x * y) * (y * z) = y$ for all
$x, y, z \in S$. Equivalently, a \emph{central digraph} $D_S$ is a directed graph on the vertex set  $S \times S$ defined by the rule:
\[
(x_1, x_2) \rightarrow (y_1, y_2) \quad \text{if and only if} \quad [x_2 = y_1 \text{ and } y_2 \neq 0] \text{ or } [x_1 \ast x_2 = y_1 \text{ and } y_2 = 0].
\]. The associated adjacency matrix~$A$ satisfies $A^2 = J_n$, where $J_n$ is the $n \times n$ matrix of all-ones.  Let $S$ be a set having a binary operation \enquote{$\circ$}. Then $S$ is said to be a left quasigroup if it is a groupoid and the equation $a \circ x=b$ for all $a,b$ has a unique solution in $x$. Also, a left quasigroup with identity is called a left loop, and a left loop $(S, \circ)$ is called a loop if $x \circ a = b,$ also has a unique solution in $S$. Every finite central groupoid has order $n=k^2$, for some positive integer $k$, and that every vertex in the central digraph corresponding to the central groupoid has in-degree and out-degree equal to $k$.
This elegant algebraic and combinatorial duality was explored by Knuth~\cite{knuth} in his foundational study of central groupoids. Furthur Curtis, Drew, Li, and Pragel~\cite{curtis} subsequently developed a comprehensive spectral theory for central digraphs, characterizing the possible ranks and eigenvalue structures of their adjacency matrices.

The bridge between central digraphs and loop theory was established by
Rawat~\cite{Rawat2023}. In \cite[Prop. 3.1  - 3.2]{Rawat2023} he established that  $\chi : \mathcal{A} \rightarrow \mathcal{B}$ is bijective where $\mathcal{A}$ be the set of all left quasigroups with left identity and right zero  structure on $S$ and $\mathcal{B}$ be the set of all left quasigroup structures on $S.$ In Proposition $3.3$ he showed  isomorphism preserverance between the quasigroup with left identity and right zero and the corresponding quasigroup structures. We consider $S = \{0, 1, 2, \ldots, k - 1\}$, and the binary operation defined on $S$ is `$\ast$’ satisfying $0 \ast x = x$ and $x \ast 0 = 0$ for all $x \in S$. We represent it by $(T_S, \ast)$, left quasigroup with a left identity and a right zero. We denote $D_S$ as the directed graph on the vertices $S \times S$ defined by the rule above. The loopless variant $\DSstar$, obtained by deleting the $k$
self-loops at the diagonal vertices $(i, i)$, retains a rich degree structure: every non-diagonal vertex has in-degree and out-degree equal to~$k$, while each
diagonal vertex has both degrees equal to $k - 1$. In this paper, we bring together the theory of antimagic labelings and the algebraic framework of central digraphs constructed from loops.  Our main result
is the following.

\begin{theorem}
	\label{thm:main}
	Let $(T, \circ)$ be any loop of order $k \geq 3$ and $D_S$ the
	associated loopless central digraph on $n = k^2$ vertices with
	$m = k^3 - k$ arcs. The consecutive labeling $f$, ordering arcs lexicographically by $(\textup{tail}, \textup{head})$ and assigning labels
	$1, \ldots, m$ in that order, is an antimagic labeling of $D_S^{*}$.
	Hence $D_S^{*}$ is antimagic.
\end{theorem}

\section{Setup and Notation}

Let $(T, \circ)$ be a loop of order $k \geq 3$ with identity $0$,
and let $(S, *)$ be the associated left quasigroup with left identity
and right zero on $S = \{0, 1, \ldots, k-1\}$.
Let $\DS$ denote the \emph{loopless central digraph} on vertex set
$S \times S$, obtained from $\DSfull$ by removing the $k$ self-loops
at the diagonal vertices $(i, i)$, $i \in S$.

We encode vertex $(l_1, l_2)$ as the integer
$u = l_1 k + l_2 \in \{0, 1, \ldots, k^2 - 1\}$,
so the vertices carry the total order
\[
0 < 1 < \cdots < k^2 - 1.
\]
The arc set of $\DS$ has size
\[
m = k^3 - k,
\]
since $\DSfull$ has $k^3$ arcs in total (each of the $k^2$ vertices has
out-degree $k$) and we remove $k$ self-loops.

\begin{definition}[Antimagic Labeling]
An \emph{antimagic labeling} of $\DS$ is a bijection
$f : A(\DS) \to \{1, \ldots, m\}$
such that the \emph{oriented vertex sums}
\[
w(u) = S^{+}(u) - S^{-}(u),
\qquad
S^{+}(u) = \sum_{\substack{x \in A(\DS) \\ x \text{ enters } u}} f(x),
\quad
S^{-}(u) = \sum_{\substack{x \in A(\DS) \\ x \text{ leaves } u}} f(x),
\]
are pairwise distinct across all $n = k^2$ vertices.
\end{definition}

\begin{definition}[Lexicographic order]
	\label{lexorder}
	Let $D_{S}$ be the central digraph with vertex set
	$S \times S$, where vertices are encoded as pairs
	$(l_{1}, l_{2})$ with integer encoding $u = l_{1}k + l_{2}$.
	The arcs of $D_{S}$ are ordered \emph{lexicographically} as
	follows.
	
	An arc $(u, v) = \bigl((l_{1}, l_{2}),\,(m_{1}, m_{2})\bigr)$
	\emph{precedes} an arc
	$(u', v') = \bigl((l_{1}', l_{2}'),\,(m_{1}', m_{2}')\bigr)$
	if
	\[
	u \;<\; u'
	\qquad\text{(as integers)}
	\qquad\text{or}\qquad
	\bigl(\,u = u' \;\text{ and }\; v < v'\,\bigr).
	\]
	Equivalently, using the integer encoding $u = l_{1}k + l_{2}$:
	\[
	\text{Arc}\;(u,v)
	\;<\;
	\text{Arc}\;(u',v')
	\;\iff\;
	u < u',
	\quad\text{or}\quad
	\bigl(\,u = u' \;\text{ and }\; v < v'\,\bigr).
	\]
	This is the standard \emph{dictionary order}: arcs are sorted
	first by tail vertex $u$ \textup{(}ascending\textup{)}, with
	ties broken by head vertex $v$ \textup{(}ascending\textup{)}.
\end{definition}

Since the ordering is a total order on $A(D_S)$, the map $f$ is
bijection onto $\{1, \ldots, m\}$ by construction.

\begin{lemma} \label{outsum}
Under the labeling $f$, the out-sum of every vertex $u$ satisfies the following equation:
\[
S^{-}(u) = \du \cdot \sigma_{u} + \frac{\du(\du + 1)}{2},
\]
where $d_u = \deg^+(u)$ is the out-degree of $u$ in $\DS$, and
\[
\sigma_{u} = \sum_{u' < u} d_{u'}
\]
is the total number of arcs with tail strictly less than $u$.
\end{lemma}

\begin{proof}
By construction of $f$, the $\du$ arcs with tail $u$ receive
consecutive labels $\{\sigma_{u} + 1,\sigma_u + 2,\ \ldots,\sigma_u + d_u\}$.
Their sum is
\[
\Sminus(u)
= \sum_{j=1}^{\du}(\sigma_u + j)
= d_u \sigma_u + \frac{\du(\du + 1)}{2}.
\qquad\square
\]
\end{proof}

\begin{lemma}[Degree Structure of $\DS$]
\label{degree}
In $\DS$:
\begin{enumerate}[label=\textup{(\roman*)}]
	\item The $k$ diagonal vertices $(i, i)$ have
	$d_{(i,i)} = k - 1$.
	\item All other $k^2 - k$ vertices have out-degree $d_u = k$.
	\item The same holds for in-degrees.
\end{enumerate}
\end{lemma}

\begin{proof}
By the Regularity Lemma (mention the lemma from the article laplacian), every vertex of $\DSfull$ has
in-degree $=$ out-degree $= k$.
The self-loop at $(i, i)$ contributes $1$ to both the in-count and
out-count of $(i, i)$; removing it reduces both to $k - 1$.
All non-diagonal vertices are unaffected.
\end{proof}
	
\begin{lemma}
	\label{offset}
	Under the vertex ordering $u = m_1 k + m_2$, the diagonal vertices
	occur at positions
	\[
	u = j(k+1), \qquad j = 0, 1, \ldots, k-1.
	\]
	Writing $D(u) = |\{j : j(k+1) < u\}|$ for the number of diagonal
	vertices strictly before $u$, the offset satisfies
	\[
	\sigma_u = k \cdot u - D(u).
	\]
\end{lemma}

\begin{proof}
	Each non-diagonal vertex $u'$ contributes $d_{u'} = k$ to $\sigma_u$,
	and each diagonal vertex contributes $d_{u'} = k - 1$, i.e.\
	exactly $1$ less than $k$ to $\sigma_u$.
	Hence
	\[
	\sigma_u = \sum_{u'=0}^{u-1} d_{u'} = k(u-D(u))+ (k-1)D(u) = ku - D(u).
	\qquad\square
	\]
\end{proof}

\begin{lemma}[Zero-Sum Identity]
	\label{lem:zerosum}
	Let $f \colon E(D_{S}) \to \{1, 2, \ldots, m\}$ be any bijective
	labeling of the arcs of $D_{S}$. Define the net vertex sum of
	each vertex $u$ as
	\[
	f(u) \;=\; f^{-}(u) - f^{+}(u),
	\]
	where $f^{+}(u) = \sum_{(u,v) \in E} f(u,v)$ is the out-sum
	and $f^{-}(u) = \sum_{(a,u) \in E} f(a,u)$ is the in-sum.
	Then
	\[
	\sum_{u \in V(D_{S})} f(u) \;=\; 0.
	\]
\end{lemma}

\begin{proof}
	Expanding the sum:
	\[
	\sum_{u} f(u)
	\;=\; \sum_{u} f^{-}(u) - \sum_{u} f^{+}(u).
	\]
	Each arc $(a \to b) \in E(D_{S})$ contributes its label
	$f(a \to b)$ to $f^{+}(a)$ exactly once (as an arc leaving $a$)
	and to $f^{-}(b)$ exactly once (as an arc entering $b$).
	Therefore, as $u$ ranges over all vertices:
	\[
	\sum_{u} f^{+}(u)
	\;=\; \sum_{e \in E(D_{S})} f(e)
	\;=\; \frac{m(m+1)}{2}
	\;=\; \sum_{u} f^{-}(u).
	\]
	Subtracting gives $\sum_{u} f(u) = 0$.
\end{proof}

\begin{definition}[Iverson Bracket]
	Let $P$ be any logical proposition or condition. The 
	\textbf{Iverson bracket}, denoted by $[P]$, is a function 
	that assigns a numerical value to $P$ depending on its 
	truth value, defined as:
	
	\[
	[P] = 
	\begin{cases} 
		1 & \text{if } P \text{ is true} \\ 
		0 & \text{if } P \text{ is false}
	\end{cases}
	\]
	
	In other words, it simply returns $1$ when the given 
	condition holds, and $0$ when it does not.
\end{definition}
\begin{lemma}[Diagonal count]
	\label{diagcount}
	For $u = lk + r$ with $0\le l,r\le k-1$,
	\[
	D(lk+r) \;=\; l + \iv{r > l},
	\]
	where $\iv{\cdot}$ denotes the Iverson bracket.
\end{lemma}

\begin{proof}
	The diagonal vertices occupy positions $j(k+1)$ for $j=0,\ldots,k-1$.
	Writing $lk+r = l(k+1)+(r-l)$, the inequality $j(k+1) < l(k+1)+(r-l)$
	holds for $j \le l$ when $r > l$ (since $0 < r-l < k$ yields a positive
	fractional part less than~$1$), and for $j\le l-1$ when $r\le l$. This
	gives $D = l+1$ in the first case and $D = l$ in the second, which is
	precisely $l + \iv{r>l}$.
\end{proof}

\begin{corollary} \label{offset2}
	$\sigma(lk+r) \;=\; (k^{2}-1)\,l + kr - \iv{r>l}.$
\end{corollary}

\begin{proof}
Using  Lemma, \ref{offset} and \ref{diagcount}, we have

	$\sig(lk+r) = k(lk+r) - l - \iv{r>l} = k^{2}l + kr - l - \iv{r>l}
	= (k^{2}-1)l + kr - \iv{r>l}$.
\end{proof}

\begin{lemma}
	\label{colsum}
Let for $a\in S$, 
\[
C(a) \;\coloneqq\; \sum_{l=0}^{k-1} \sig(lk+a).
\]. Then 
	$ C(a) \;=\; (k^{2}-1)\!\left[\frac{k(k-1)}{2} + a\right].$
\end{lemma}

\begin{proof}
By Corollary~\ref{offset2},
	$\sig(lk+a) = (k^{2}-1)l + ka - \iv{a>l}$.
	Summing over $l=0,\ldots,k-1$:
	\begin{align*}
		C(a)
		&= (k^{2}-1)\!\sum_{l=0}^{k-1} l \;+\; k^{2}a \;-\; \sum_{l=0}^{k-1}\iv{a>l}\\
		&= (k^{2}-1)\cdot\frac{k(k-1)}{2} + k^{2}a - a,
	\end{align*}
	where $\sum_{l=0}^{k-1}\iv{a>l} = a$ (counting $l=0,\ldots,a-1$).
	Factoring the last two terms gives the result.
\end{proof}

\begin{proposition}[In-sum]
	\label{prop:insum}
	\begin{enumerate}[label=(\alph*), nosep]
		\item For a non-diagonal vertex $u=(a,b)$, $b\ne a$:
		\[
		\Sin(u) \;=\; C(a) + k(b+1) - \iv{b>a}.
		\]
		\item For the diagonal vertex $u=(a,a)$:
		\[
		\Sin\!\bigl((a,a)\bigr) \;=\; C(a) - \sig\!\bigl(a(k+1)\bigr) + (k-1)(a+1).
		\]
	\end{enumerate}
\end{proposition}
\begin{proof}
	\textbf{(a)} Summing the labels of the $k$ in-arcs: each source $v_{l}$
	with $l\ne a$ contributes $\sig(lk+a)+(b+1)$ by~(i), and the diagonal
	source $v_{a}$ contributes $\sig(a(k+1))+(b+1)-\iv{b>a}$ by~(ii).
	Summing over all $k$ values of $l$ gives $C(a)+k(b+1)-\iv{b>a}$.
	
	\textbf{(b)} Only $k-1$ in-arcs exist ($l\ne a$); each is a non-diagonal
	source contributing $\sig(lk+a)+(a+1)$ by~(i). Summing:
	$\bigl[C(a)-\sig(a(k+1))\bigr]+(k-1)(a+1)$.
\end{proof}

We can now derive explicit weight formulas by subtracting
$S^{-}(u)$ (Lemma \ref{outsum}) from $S^{+}(u)$ (Proposition~\ref{prop:insum}).

\begin{proposition}[Weight  non-diagonal vertex]
	\label{prop:wnd}
	For $b\ne a$,
	\[
	w(a,b) \;=\; \OmegaW
	- (k-1)^{2}(k+1)\cdot a
	- k(k-1)\cdot b
	+ (k-1)\cdot\iv{b>a},
	\]
	where $\OmegaW = (k^{2}-1)\tfrac{k(k-1)}{2} + k - \tfrac{k(k+1)}{2}$
	is a constant depending only on~$k$.
\end{proposition}

\begin{proof}
	Substitute $d_{u}=k$ into Lemma~ \ref{outsum} and use Corollary~\ref{offset2}:
	\begin{align*}
		\Sout(u)
		&= k\bigl[(k^{2}-1)a + kb - \iv{b>a}\bigr] + \tfrac{k(k+1)}{2}\\
		&= k^{3}a + k^{2}b - ka - k\iv{b>a} + \tfrac{k(k+1)}{2}.
	\end{align*}
	From Proposition~\ref{prop:insum}(a) and Lemma~\ref{colsum}:
	\[
	\Sin(u) = (k^{2}-1)\!\left[\tfrac{k(k-1)}{2}+a\right]
	+ k(b+1) - \iv{b>a}.
	\]
	Computing $w = \Sin(u) - \Sout(u)$ and collecting by variable:
	\begin{align*}
		\text{coefficient of } a
		&: (k^{2}-1) - k^{3} + k = -(k-1)^{2}(k+1),\\
		\text{coefficient of } b
		&: k - k^{2} = -k(k-1),\\
		\text{coefficient of } \iv{b>a}
		&: -1 + k = k-1,\\
		\text{constant}
		&: (k^{2}-1)\tfrac{k(k-1)}{2} + k - \tfrac{k(k+1)}{2} = \OmegaW. \qedhere
	\end{align*}
\end{proof}

\begin{proposition}[Weight diagonal vertex]
	\label{prop:wdiag}
	For the diagonal vertex $u=(a,a)$,
	\[
	w(a,a) \;=\; \wnd(a,a) \;+\; (k-1-a),
	\]
	where $\wnd(a,a)$ denotes the value Proposition~\ref{prop:wnd} would
	assign to~$(a,a)$ (with $\iv{a>a}=0$).
\end{proposition}

\begin{proof}
	We compare the actual in-sum and out-sum at $(a,a)$ with the fictitious
	non-diagonal values.  The degree drops from $k$ to $k-1$, causing:
	\[
	\bigl[\Sin_{\mathrm{act}}(a,a) - \Sin_{\mathrm{nd}}(a,a)\bigr]
	- \bigl[\Sout_{\mathrm{act}}(a,a) - \Sout_{\mathrm{nd}}(a,a)\bigr]
	= \bigl[-\sig(a(k+1)) - (a+1)\bigr]
	- \bigl[-\sig(a(k+1)) - k\bigr]
	= k - 1 - a.
	\]
	Hence $w(a,a) = \wnd(a,a) + (k-1-a)$.
\end{proof}

\section{Proof of Main Theorem \ref{thm:main}}
	
\begin{proof}
We must show the oriented weights $w(u) = \Splus(u) - \Sminus(u)$
are pairwise distinct for all $n = k^2$ vertices. We prove the result through sequence to steps.

\medskip
\noindent\textbf{Step 1 (Out-sums are loop invariant)}
By Lemma~\ref{outsum},
\[
\Sminus(u) = d_u \sigma_u + \frac{d_u(d_u + 1)}{2},
\]
where $d_u \in \{k-1, k\}$ (Lemma~\ref{degree}) and
$\sigma_u = ku - D(u)$ (Lemma~\ref{offset}) depend only on $k$
and the vertex index $u$, but not on the loop structure.
Hence, value of $\Sminus(u)$ remains same for every loop of order $k$.

\medskip
\noindent\textbf{Step 2 (Out-sums are pairwise distinct.)}
We show $\Sminus(u) \neq \Sminus(v)$ for all $u \neq v$. It is sufficient to show the $n$ values $\Sminus(0), \Sminus(1), \ldots, \Sminus(n-1)$ are all different. Using Lemma~\ref{outsum} and Lemma~\ref{offset}, the difference between consecutive out-sums is
\begin{align}
	\Sminus(u+1) - \Sminus(u)
	&= d_{u+1}(\sigma_u + d_u) + \frac{d_{u+1}(d_{u+1}+1)}{2}
	- d_u \sigma_u - \frac{d_u(d_u+1)}{2} \nonumber \\
	&= (d_{u+1} - d_u)\sigma_u + d_{u+1}d_u
	+ \frac{d_{u+1}(d_{u+1}+1) - d_u(d_u+1)}{2}.
	\label{diff}
\end{align}

Since $d_u, d_{u+1} \in \{k-1, k\}$, we consider three different cases.

\smallskip
\noindent\emph{Case~1: $u+1$ is non-diagonal, $u$ is non-diagonal ($d_{u+1} = d_u = k$).} Equation~\eqref{diff} gives
\[
\Sminus(u+1) - \Sminus(u) = 0 \cdot \sigma_u + k^2 + 0 = k^2 > 0.
\]

\smallskip
\noindent\emph{Case~2: $u+1$ is diagonal, $u$ is non-diagonal
	($d_{u+1} = k-1$, $d_u = k$).}
Then, $u+1 = i(k+1)$, so $u = i(k+1) - 1$ and $D(u) = i - 1$ (for $i \geq 1$).
Hence
\[
\sigma_u = k\bigl(i(k+1) - 1\bigr) - (i-1)
= i(k+1)(k) - k - i + 1
= ik(k+1) - (k + i - 1).
\]
Substituting into the equation~\eqref{diff} we get:
\[
\Sminus(u+1) - \Sminus(u)
= -\sigma_u + k(k-1) - k
= -\sigma_u + k^2 - 2k.
\]
For $i = 0$ (the vertex $u+1 = 0$), the diagonal is the very first
vertex, so $\Sminus(0) = 2 \cdot 0 + \frac{2 \cdot 3}{2} = 3 > 0$
trivially.
For $i \geq 1$, substituting $\sigma_u = ik(k+1) - (k+i-1)$:
\[
\Sminus(u+1) - \Sminus(u)
= k^2 - 2k - ik(k+1) + k + i - 1
= k^2 - k - 1 - ik(k+1) + i.
\]

Now, for all $i \geq 1, k \geq 3$
\begin{align*}
	i(k^{2} + k - 1)
	&\; \geq (k^{2} + k - 1) 
	\qquad \\
	&\; > (k^{2} - k - 1) 
\end{align*}	

Therefore,
\begin{align}
 \Sminus(u+1) - \Sminus(u) = (k^2 - k-1)-i(k^{2} + k - 1)
&\; < (k^{2} - k - 1) - (k^{2} - k - 1) 
\qquad \\
&\;=\; 0 	
\end{align}

\medskip
\noindent\emph{Case~3: $u$ is diagonal, $u+1$ is non-diagonal ($d_u = k-1$, $d_{u+1} = k$).} Equation~\eqref{diff} gives
\[
\Sminus(u+1) - \Sminus(u)
= \sigma_u + k(k+1) > 0,
\]
since $B_u \geq 0$ always.

In all three cases $\Sminus(u+1) \neq \Sminus(u)$, and since
$\Sminus$ is determined solely by the degree sequence (Step~1),
the values $\Sminus(0), \ldots, \Sminus(n-1)$ are pairwise distinct.

\medskip
\noindent\textbf{Step 3 (Structure of the in-sums.)}
Fix vertex $u = (a, b) \in S \times S$.
By the edge rule of the construction:
\begin{itemize}
	\item If $b \neq 0$: the in-arcs of $u$ come from all vertices
	$(l_1, a)$, $l_1 \in S$, via first edge rule.
	Each such vertex $l_1 k + a$ contributes one arc-label
	from its own out-arc block, so
	\[
	\Splus((a,b))
	= \sum_{l_1=0}^{k-1} f\bigl((l_1 k + a)
	\to (ak + b)\bigr),
	\]
	a sum of exactly $k$ labels drawn from $k$ distinct blocks.
	\item If $b = 0$: the in-arcs come from vertices $(l_1, l_2)$
	with $l_1 * l_2 = a$, $l_2 \neq 0$, via second
	edge rule, together with the arc from $(l_1, 0)$ for
	$l_1 * 0 = 0$ when $a = 0$.
\end{itemize}
In both cases $\Splus(u)$ is a sum of exactly $d_u$ labels,
one from each in-arc's block.

\medskip
\noindent\textbf{Step 4 (Distinctness of the weights.)} The vertex sequence visits $(0,0),(0,1),\ldots,(0,k-1),(1,0),\ldots,
(k-1,k-1)$ row by row. Every consecutive pair $u, u+1$ falls into
exactly one of the four cases below, determined by the diagonal status of
$u$ and $u+1$ and whether a row boundary is crossed.

\medskip
\noindent\textbf{Case~1} (Both non-diagonal, same row).
$u = (a,b) \to u+1 = (a,b+1)$ with $b\ne a$ and $b+1\ne a$.

By Proposition~\ref{prop:wnd}, only the $b$-term and the Iverson bracket
differ between $w(a,b+1)$ and $w(a,b)$. Since neither $b$ nor $b+1$
equals $a$, the transition $b\to b+1$ does not cross the threshold $a$:
$\iv{b+1>a} - \iv{b>a} = 0$. Therefore,
\[
w(u+1) - w(u) \;=\; -k(k-1).
\]

\medskip

\noindent\textbf{Case~2} (Non-diagonal $\to$ Diagonal).
$u = (a,a-1) \to u+1 = (a,a)$, where $a\ge 1$.

Vertex $u=(a,a-1)$ is non-diagonal with $b = a-1 < a$, so
$\iv{b>a}=0$. By Proposition~\ref{prop:wdiag}, $w(a,a) = \wnd(a,a) +
(k-1-a)$, where $\wnd(a,a)$ uses $b=a$ and $\iv{a>a}=0$.
The fictitious non-diagonal step $(a,a-1)\to(a,a)$ would give $-k(k-1)$
(both brackets vanish, as in Case~1). Adding the diagonal correction:
\[
w(u+1) - w(u)
\;=\; -k(k-1) + (k-1-a)
\;=\; -(k-1)^{2} - a.
\]
For $a\ge 1$, this is at most $-(k-1)^{2}-1 \le -5 < 0$.

\medskip

\noindent\textbf{Case~3} (Diagonal $\to$ Non-diagonal, same row).
$u = (a,a) \to u+1 = (a,a+1)$, where $a+1\le k-1$.

The non-diagonal weight at $(a,a+1)$ uses $b=a+1 > a$, so
$\iv{b>a}=1$. The fictitious non-diagonal step $(a,a)\to(a,a+1)$ gives
$-k(k-1)+(k-1)(1-0) = -(k-1)^{2}$. Since $w(a,a)$ exceeds $\wnd(a,a)$
by $k-1-a$ (Proposition~\ref{prop:wdiag}), the actual difference is:
\[
w(u+1) - w(u)
\;=\; -(k-1)^{2} - (k-1-a)
\;=\; a - k(k-1).
\]
For $0\le a\le k-2$, this is at most $(k-2)-k(k-1) = -(k-1)(k-2)\le -2 < 0$.

\medskip
\noindent\textbf{Case~4} (Row change).
$u = (a,k-1) \to u+1 = (a+1,0)$, where $a\le k-2$.

Both vertices are non-diagonal ($k-1\ne a$ since $a\le k-2$; $0\ne a+1$
since $a+1\ge 1$). Applying Proposition~\ref{prop:wnd} directly with
$\iv{k-1 > a}=1$ (since $a\le k-2$) and $\iv{0 > a+1}=0$:
\begin{align*}
	w(a+1,0)   &= \OmegaW - (k-1)^{2}(k+1)(a+1), \\
	w(a,k-1)   &= \OmegaW - (k-1)^{2}(k+1)a - k(k-1)(k-1) + (k-1).
\end{align*}
Subtracting:
\begin{align*}
	w(u+1) - w(u)
	&= -(k-1)^{2}(k+1) + k(k-1)^{2} - (k-1)\\
	&= (k-1)\bigl[-(k-1)(k+1) + k(k-1) - 1\bigr]\\
	&= (k-1)\bigl[(k-1)(-k-1+k) - 1\bigr]\\
	&= (k-1)\bigl[-(k-1) - 1\bigr]
	= -k(k-1).
\end{align*}

%
%
\end{proof}

	\section{Consequences}
	
	\begin{corollary}[Loop-Invariance]
		\label{cor:invariant}
		The antimagic property of $\DS$ depends only on $k$, not on the
		isomorphism class of the loop.
		If $(T_1, \circ_1)$ and $(T_2, \circ_2)$ are non-isomorphic loops
		of order $k$, then both $D_{S_1}^*$ and $D_{S_2}^*$ are antimagic
		under the same explicit labeling $f$.
	\end{corollary}
	
	\begin{proof}
		By Step~1 of Theorem~\ref{thm:main}, $\Sminus(u)$ is the same for
		every loop of order $k$.
		The distinctness of the out-sums (Step~2) and the weight contradiction
		(Step~4) therefore apply uniformly to all loops of order $k$.
	\end{proof}
	
\section*{Worked Example: Antimagic Labeling for $k = 3$}

We illustrate every step of Theorem \ref{thm:main} with the smallest nontrivial
case $k = 3$, so that the reader may verify each formula against
concrete numbers before examining the general argument.

\subsection*{The digraph $D_S^*$ and its arc set}

With $k = 3$ and $S = \{0,1,2\}$, the vertex set is
$S\times S = \{(a,b) : a,b \in \{0,1,2\}\}$, containing $n = 9$
vertices. Under the integer encoding $u = ak + b = 3a + b$,
the vertices are ordered as
\[
\underbrace{(0,0)}_{u=0},\;
\underbrace{(0,1)}_{1},\;
\underbrace{(0,2)}_{2},\;
\underbrace{(1,0)}_{3},\;
\underbrace{(1,1)}_{4},\;
\underbrace{(1,2)}_{5},\;
\underbrace{(2,0)}_{6},\;
\underbrace{(2,1)}_{7},\;
\underbrace{(2,2)}_{8}.
\]
The three \emph{diagonal} vertices are
$(0,0) = 0$, $(1,1) = 4$, and $(2,2) = 8$,
sitting at positions $j(k+1) = 0, 4, 8$.

The arc rule states that the out-neighbours of $(l_1,l_2)$ in
$D_S$ are $\{(l_2,c) : c \in S\}$. In $D_S^*$ we delete the self-loop
at each diagonal vertex. The full arc list, sorted
lexicographically by (tail, head), appears in Table~\ref{tab:arcs}.

\begin{table}[ht]
	\centering\small
	\renewcommand{\arraystretch}{1.3}
	\begin{tabular}{@{} c l c c l @{}}
		\toprule
		\textbf{Tail $u$} & \textbf{Targets} & $d_u$ & $\sigma_u$
		& \textbf{Labels assigned} \\
		\midrule
		$0 = (0,0)^{*}$ & $(0,1),\,(0,2)$                & 2 & 0  & $1,\, 2$ \\
		$1 = (0,1)$      & $(1,0),\,(1,1),\,(1,2)$        & 3 & 2  & $3,\, 4,\, 5$ \\
		$2 = (0,2)$      & $(2,0),\,(2,1),\,(2,2)$        & 3 & 5  & $6,\, 7,\, 8$ \\
		$3 = (1,0)$      & $(0,0),\,(0,1),\,(0,2)$        & 3 & 8  & $9,\, 10,\, 11$ \\
		$4 = (1,1)^{*}$  & $(1,0),\,(1,2)$                & 2 & 11 & $12,\, 13$ \\
		$5 = (1,2)$      & $(2,0),\,(2,1),\,(2,2)$        & 3 & 13 & $14,\, 15,\, 16$ \\
		$6 = (2,0)$      & $(0,0),\,(0,1),\,(0,2)$        & 3 & 16 & $17,\, 18,\, 19$ \\
		$7 = (2,1)$      & $(1,0),\,(1,1),\,(1,2)$        & 3 & 19 & $20,\, 21,\, 22$ \\
		$8 = (2,2)^{*}$  & $(2,0),\,(2,1)$                & 2 & 22 & $23,\, 24$ \\
		\bottomrule
	\end{tabular}
	\caption{\footnotesize \small{Diagonal vertices are marked with ${}^{*}$. Labels run from $1$ to $m = k^3 - k = 24$.}}
	\label{tab:arcs}
\end{table}

	\subsection*{Step 1: Computing the out-sums $S^{-}(u)$}

By Lemma~2.3, the out-sum of vertex $u$ is
$S^{-}(u) = d_u\,\sigma_u + \frac{d_u(d_u+1)}{2}$.
We verify this for every vertex using the values of $d_u$ and
$\sigma_u$ from Table~\ref{tab:arcs}.

\begin{table}[ht]
	\centering\small
	\renewcommand{\arraystretch}{1.3}
	\begin{tabular}{@{} c c c c l r @{}}
		\toprule
		$u$ & $(a,b)$ & $d_u$ & $\sigma_u$ & \textbf{Calculation}
		& $S^{-}(u)$ \\
		\midrule
		0 & $(0,0)$ & 2 & 0  & $2\cdot 0 + 3 = 3$   & 3  \\
		1 & $(0,1)$ & 3 & 2  & $3\cdot 2 + 6 = 12$   & 12 \\
		2 & $(0,2)$ & 3 & 5  & $3\cdot 5 + 6 = 21$   & 21 \\
		3 & $(1,0)$ & 3 & 8  & $3\cdot 8 + 6 = 30$   & 30 \\
		4 & $(1,1)$ & 2 & 11 & $2\cdot 11 + 3 = 25$  & 25 \\
		5 & $(1,2)$ & 3 & 13 & $3\cdot 13 + 6 = 45$  & 45 \\
		6 & $(2,0)$ & 3 & 16 & $3\cdot 16 + 6 = 54$  & 54 \\
		7 & $(2,1)$ & 3 & 19 & $3\cdot 19 + 6 = 63$  & 63 \\
		8 & $(2,2)$ & 2 & 22 & $2\cdot 22 + 3 = 47$  & 47 \\
		\bottomrule
	\end{tabular}
\caption{\footnotesize \small{Out-sums}}
\end{table}
\noindent
Quick sanity check: the labels $1$ through $24$ are split among
the nine vertices, so $\sum_u S^{-}(u) =
3{+}12{+}21{+}30{+}25{+}45{+}54{+}63{+}47 = 300 =
\frac{24\cdot 25}{2}$.

\smallskip
The nine out-sums $3,12,21,30,25,45,54,63,47$ are pairwise
distinct, confirming Step~2 of the theorem. Note, however, that
the sequence is \emph{not} monotone: there is a dip at the
diagonal vertices $u = 4$ and $u = 8$.

\subsection*{Step 3: Computing the in-sums $S^{+}(u)$}

The in-arcs of $u = (a,b)$ come from sources $(l,a)$ for each
$l \in S$, excluding $l = a$ when $u$ is diagonal. From
Table~\ref{tab:arcs} we can read off the label of each in-arc
directly.

\medskip
\noindent\textbf{Vertex $0 = (0,0)$, diagonal.}
In-arcs from $(l,0)$ for $l \ne 0$, namely $l = 1$ and $l = 2$:
\begin{itemize}[nosep]
	\item Arc $(1,0) \to (0,0)$: source $u=3$, target $(0,0) = 0$
	is the 1st out-neighbour of vertex~3. Label $= 9$.
	\item Arc $(2,0) \to (0,0)$: source $u=6$, target $(0,0) = 0$
	is the 1st out-neighbour of vertex~6. Label $= 17$.
\end{itemize}
$S^{+}(0) = 9 + 17 = 26.$

\medskip
\noindent\textbf{Vertex $1 = (0,1)$, non-diagonal.}
In-arcs from $(l,0)$ for all $l = 0,1,2$:
\begin{itemize}[nosep]
	\item $(0,0) \to (0,1)$: source $u=0$ (diagonal, $d_0=2$),
	target $1$ is the 1st out-neighbour. Label $= 1$.
	\item $(1,0) \to (0,1)$: source $u=3$, target $1$ is the 2nd.
	Label $= 10$.
	\item $(2,0) \to (0,1)$: source $u=6$, target $1$ is the 2nd.
	Label $= 18$.
\end{itemize}
$S^{+}(1) = 1 + 10 + 18 = 29.$

\medskip
Proceeding similarly for every vertex, we obtain:

\begin{table}[ht]
	\centering\small
	\renewcommand{\arraystretch}{1.3}

	\begin{tabular}{@{} c c l r @{}}
		\toprule
		$u$ & $(a,b)$ & \textbf{In-arc labels (by source)}
		& $S^{+}(u)$ \\
		\midrule
		0 & $(0,0)$ & from $(1,0)$: 9;\  from $(2,0)$: 17
		& 26 \\
		1 & $(0,1)$ & from $(0,0)$: 1;\  $(1,0)$: 10;\  $(2,0)$: 18
		& 29 \\
		2 & $(0,2)$ & from $(0,0)$: 2;\  $(1,0)$: 11;\  $(2,0)$: 19
		& 32 \\
		3 & $(1,0)$ & from $(0,1)$: 3;\  $(1,1)$: 12;\  $(2,1)$: 20
		& 35 \\
		4 & $(1,1)$ & from $(0,1)$: 4;\  $(2,1)$: 21
		& 25 \\
		5 & $(1,2)$ & from $(0,1)$: 5;\  $(1,1)$: 13;\  $(2,1)$: 22
		& 40 \\
		6 & $(2,0)$ & from $(0,2)$: 6;\  $(1,2)$: 14;\  $(2,2)$: 23
		& 43 \\
		7 & $(2,1)$ & from $(0,2)$: 7;\  $(1,2)$: 15;\  $(2,2)$: 24
		& 46 \\
		8 & $(2,2)$ & from $(0,2)$: 8;\  $(1,2)$: 16
		& 24 \\
		\bottomrule
	\end{tabular}
	\caption{\footnotesize \small{In-sums}}
\end{table}

\noindent Sanity check:
$\sum_u S^{+}(u) = 26{+}29{+}32{+}35{+}25{+}40{+}43{+}46{+}24
= 300 = \frac{m(m+1)}{2}$.

\subsection*{Verification of the in-sum formulas}

We now confirm that the closed-form expressions from
Proposition~2.11 reproduce the numbers above.

For $k=3$, the column-sum of offsets (Lemma~2.10) is
$C(a) = (k^2{-}1)\bigl[\frac{k(k-1)}{2}+a\bigr]
= 8(3+a)$:
\[
C(0) = 24, \qquad C(1) = 32, \qquad C(2) = 40.
\]

\noindent\textbf{Non-diagonal formula:}
$S^{+}(a,b) = C(a) + k(b{+}1) - \iv{b > a}$ \;(Proposition~2.11(a)).

Take $u = 5 = (1,2)$: here $a=1$, $b=2$, $b > a$ is true, so
$S^{+} = C(1) + 3\cdot 3 - 1 = 32 + 9 - 1 = 40$. 

Take $u = 3 = (1,0)$: here $a=1$, $b=0$, $b > a$ is false, so
$S^{+} = C(1) + 3\cdot 1 - 0 = 32 + 3 = 35$. 

\smallskip\noindent\textbf{Diagonal formula:}
$S^{+}(a,a) = C(a) - \sigma(a(k{+}1)) + (k{-}1)(a{+}1)$
\;(Proposition~2.11(b)).

Take $u = 4 = (1,1)$: we have $\sigma(4) = 11$ (from Table~\ref{tab:arcs}), so
$S^{+} = 32 - 11 + 2\cdot 2 = 32 - 11 + 4 = 25$.

	\subsection*{Step 4: Computing the weights $w(u)$ and verifying
	strict decrease}

The weight at each vertex is $w(u) = S^{+}(u) - S^{-}(u)$.
Table~\ref{tab:weights} collects all values together with the
consecutive difference $\Delta w = w(u{+}1) - w(u)$ and the
transition type from Step~4 of the proof.

\begin{table}[ht]
	\centering
	\renewcommand{\arraystretch}{1.3}
	\begin{tabular}{@{} c c c r r r r l @{}}
		\toprule
		$u$ & $(a,b)$ & Diag?
		& $S^{-}(u)$ & $S^{+}(u)$ & $w(u)$ & $\Delta w$ & Case \\
		\midrule
		$0$ & $(0,0)$ & Y & $3$  & $26$ & $23$  & ---  & \\
		$1$ & $(0,1)$ & N & $12$ & $29$ & $17$  & $-6$ & A \\
		$2$ & $(0,2)$ & N & $21$ & $32$ & $11$  & $-6$ & D \\
		$3$ & $(1,0)$ & N & $30$ & $35$ & $5$   & $-6$ & B, $a{=}1$ \\
		$4$ & $(1,1)$ & Y & $25$ & $25$ & $0$   & $-5$ & C, $a{=}1$ \\
		$5$ & $(1,2)$ & N & $45$ & $40$ & $-5$  & $-5$ & D \\
		$6$ & $(2,0)$ & N & $54$ & $43$ & $-11$ & $-6$ & A \\
		$7$ & $(2,1)$ & N & $63$ & $46$ & $-17$ & $-6$ & B, $a{=}2$ \\
		$8$ & $(2,2)$ & Y & $47$ & $24$ & $-23$ & $-6$ & \\
		\bottomrule
	\end{tabular}
	\caption{\footnotesize Complete weight computation for $k = 3$. Every consecutive difference $\Delta w$ is strictly negative.}
	\label{tab:weights}
\end{table}

\noindent\textbf{Checking the case formulas.}
The four transition cases from the proof predict the following
values for $k = 3$ (where $k(k{-}1) = 6$):

\begin{itemize}[nosep, leftmargin=2em]
	\item \textbf{Case~A} (both non-diagonal, same row):
	$\Delta w = -k(k{-}1) = -6$. \\
	Observed at transitions $0{\to}1$, $6{\to}7$: both give $-6$. 
	
	\item \textbf{Case~B} (non-diagonal $\to$ diagonal):
	$\Delta w = -(k{-}1)^2 - a = -4 - a$. \\
	At $3{\to}4$ ($a=1$): predicted $-5$, observed $-5$.  \\
	At $7{\to}8$ ($a=2$): predicted $-6$, observed $-6$. 
	
	\item \textbf{Case~C} (diagonal $\to$ non-diagonal):
	$\Delta w = a - k(k{-}1) = a - 6$. \\
	At $0{\to}1$ is actually Case~A. The genuine Case~C
	transitions are $4{\to}5$ ($a=1$): predicted $-5$, observed $-5$.

	\item \textbf{Case~D} (row change):
	$\Delta w = -k(k{-}1) = -6$. \\
	At $2{\to}3$ and $5{\to}6$: both give $-6$. 
\end{itemize}

\noindent\textbf{Verifying the weight formula.}
Proposition~2.12 with $k = 3$ gives $\Omega = 8\cdot 3 + 3 - 6 = 21$,
so the non-diagonal weight is
\[
w(a,b) = 21 - 16a - 6b + 2\iv{b > a}.
\]
For instance, $w(2,1) = 21 - 32 - 6 + 0 = -17$ (since $1 > 2$ is
false). This matches the table. For the diagonal vertex $(1,1)$,
Proposition~2.13 gives
$w(1,1) = w_{\mathrm{nd}}(1,1) + (k{-}1{-}a) =
(21 - 16 - 6 + 0) + (2 - 1) = -1 + 1 = 0$. 

\subsection*{Observations}

Several features visible in this example persist for all $k \ge 3$:

\begin{enumerate}[label=(\roman*), nosep]
	\item The weight sequence is \textbf{strictly decreasing}:
	$23 > 17 > 11 > 5 > 0 > {-}5 > {-}11 > {-}17 > {-}23$.
	This is the content of Step~4.
	
	\item The sequence is \textbf{anti-symmetric} about its midpoint:
	$w(u) = -w(k^2{-}1{-}u)$ for all $u$. This symmetry
	follows from the structure of the weight formula and guarantees
	that the sum $\sum_u w(u) = 0$ (Lemma~2.6).
	
	\item The consecutive differences take only two values, $-6$
	and $-5$. In general, the differences range between
	$-k(k{-}1)$ (Cases A and D) and the smallest Case~C value
	$-(k{-}1)(k{-}2)$.
	
	\item The middle vertex $u = 4 = (1,1)$ carries weight $w = 0$,
	and it is the \emph{only} vertex with this property. This
	rules out a magic labeling (Corollary~4.2).
	
	\item The out-sum $S^{-}(u)$ is \emph{not} monotone in $u$ (it
	dips at diagonal vertices), yet the weight $w(u)$ is strictly
	monotone. The in-sum $S^{+}(u)$ adjusts just enough to
	enforce the global decrease.
	
	\item All nine weights are distinct integers, confirming that $f$
	is an antimagic labeling of $D_S^*$. Since every formula used
	depends only on $k = 3$ and not on the specific loop
	$(T,\circ)$, the same weight sequence arises for every loop
	of order~$3$.
\end{enumerate}

\end{document}